\documentclass[copyright,creativecommons]{eptcs}
\usepackage{breakurl}

\usepackage{amssymb,amsmath,amscd,latexsym,epic,eepic}
\usepackage{algorithm}

\newtheorem{theorem}{Theorem}

\newtheorem{corollary}{Corollary}
\newenvironment{proof}{{\bf Proof.}}{}

\def\qed{\rule{0.4em}{1.4ex}} 
 \newcommand{\set}[1]{\{#1\}}

\newcommand{\PA}{1}
\newcommand{\PB}{2}
\newcommand{\SA}{{S_1}}

\newcommand{\SB}{{S_2}}

\newcommand{\SR}{S_{P}}

\newcommand{\gamegraph}{G}

\newcommand{\winval}[1]{\langle \! \langle #1 \rangle\! \rangle_{\mathit{val}} }
\newcommand{\va}{\winval{1}}
\newcommand{\vb}{\winval{2}}

\newcommand{\Prb}{\mathrm{Pr}}
\newcommand{\Exp}{\mathbb{E}}
\newcommand{\seq}[1]{\langle #1 \rangle}

\newcommand{\trans}{\delta}
\newcommand{\distr}{{\cal D}}

\newcommand{\supp}{\mathrm{Supp}}

\newcommand{\FM}{\mathit{F}}
\newcommand{\PF}{\mathit{PF}}
\newcommand{\PM}{\mathit{PM}}
\newcommand{\mem}{{\tt M}}
\newcommand{\Aa}{{\cal A}}

\newcommand{\slopefrac}[2]{\leavevmode\kern.1em
  \raise .5ex\hbox{\the\scriptfont0 #1}\kern-.1em
  /\kern-.15em\lower .25ex\hbox{\the\scriptfont0 #2}}
\newcommand{\half}{\slopefrac{1}{2}}

\newcommand{\pat}{\omega}  \newcommand{\Paths}{\Omega}
\newcommand{\straa}{\sigma} \newcommand{\Straa}{\Sigma}
\newcommand{\strab}{\pi} \newcommand{\Strab}{\Pi}

 \newcommand{\Nats}{\mathbb{N}}
\newcommand{\reals}{\mathbb{R}} 
 
 \newcommand{\Reach}[1]{\mathrm{Reach}(#1)}

\newcommand{\TwoDisc}{\mathrm{TwoDisc}}

\makeatletter

\begingroup \catcode `|=0 \catcode `[= 1
\catcode`]=2 \catcode `\{=12 \catcode `\}=12
\catcode`\\=12 |gdef|@xcomment#1\end{comment}[|end[comment]]
|endgroup

\def\@comment{\let\do\@makeother \dospecials\catcode`\^^M=10\def\par{}}

\def\begincomment{\@comment\@xcomment}

\makeatother

\title{Discounting in Games across Time Scales\thanks{
  This research was funded in part by the US National Science Foundation grants
  CCF-0546170 and CNS-0702881, and DARPA grant HR0011-09-1-0037.
}}
\author{Krishnendu Chatterjee
\institute{IST Austria (Institute of Science and Technology Austria)}
\email{krishnendu.chatterjee@ist.ac.at}
\and
\qquad\qquad Rupak Majumdar
\institute{University of California, Los Angeles, USA and}
\institute{Max Planck Institute for Software Systems, Germany}
\email{rupak@cs.ucla.edu}
}
\def\titlerunning{Discounting in Games across Time Scales}
\def\authorrunning{K.Chatterjee and R. Majumdar}
\begin{document}
\maketitle
\begin{abstract}
We introduce two-level discounted games played by two players on a
perfect-information stochastic game graph.
The upper level game is a discounted game and the lower level game is an
undiscounted reachability game.
Two-level games model hierarchical and sequential decision making under uncertainty
across different time scales.
We show the existence of pure memoryless optimal strategies for both players and
an ordered field property for such games.
We show that if there is only one player (Markov decision processes), then
the values can be computed in polynomial time.
It follows that whether the value of a player is equal to a given rational
constant in two-level discounted games can be decided in NP $\cap$ coNP.
We also give an alternate strategy improvement algorithm to compute the value.
\end{abstract}

\section{Introduction}

Discrete stochastic games have been extensively studied
as models for decision making under adversarial interactions in
an uncertain environment, and have found many applications, such as in
manufacturing systems control and inventory management \cite{FV97}.

In many such applications, the interaction with the environment occurs
in a hierarchical manner, intercalated across different time scales.
In the short-term, a decision has to be made about choosing one
of several possible actions.
For example, short term decisions can determine whether to buy a certain
product or another, or whether to increase or decrease production capacity.
In the long-term, the system gets a profit or a loss at each step
based on its existing inventory.
Both short-term and long-term decisions can potentially involve
uncertainty and adversarial interactions.
Moreover, long term decisions are influenced by the short term actions chosen,
and model the effect of the local decisions on the overall
profits or losses of the system.

Technically,
the two types of interaction are modeled using two distinct
classes of games.
Undiscounted {\em reachability games} are used to model short-term decision
making (e.g., what action to take next).
In a reachability game on a state space, one fixes a set of goal states, and
the objective of player~1 is to maximize the probability of reaching
the goal states.

On the other hand, {\em discounted reward games}
model long-term rewards for the system
(e.g., how the actions chosen locally relate to long-term profits).
In a discounted game,
player~1 gets a reward in each step, and a time discount
parameter $\lambda \in (0,1)$
is used to ``discount'' the reward at future time points (i.e.,
the reward $r$ obtained $t$ time units in the future is given a value
$\lambda^tr$).
The short-term interactions are abstracted away into an atomic step
that uniformly sets the time granularity.
The objective of player~1 is to maximize the expected normalized sum of
discounted rewards.

To make the games concrete, consider economic policymakers setting
financial policy.
The specific policy implemented (e.g., the interest
rate or the amount of regulation) affects the long-term health of the economy,
and the interplay between financial policy and the market can be modeled using discounted rewards.
However, in each step, the specific policy chosen depends on ``short-term'' games
between various stakeholders, such as politicians, the treasury, companies, and various interest groups.
In this setting, the time granularity of long-term steps (policy implementation) is variable,
and depends on the length and outcome of the short-term steps (deciding which policy to implement).

While each game model in itself provides a sound theoretical basis
for reasoning about system behavior,
the hierarchical interaction and varying time granularities (short-term
vs.\ long-term)
present in many applications is not adequately captured by either model.
In this paper, we introduce models for such multi-level interactions and
algorithms for sequential decision making in a setting where the time granularity
can be variable.
We define a {\em two-level} discounted game, in which
a ``lower level'' reachability game is used to decide actions for
a ``higher level'' discounted game.
The discount factor is applied to the time scale of the higher-level game,
not for every step that elapses in the lower-level game.
Since every lower level game is different we obtain complete independence
in granularity of transitions.

Our main result is the existence of value and pure memoryless strategies
in two-level discounted games.
Moreover, we show that the value and optimal strategies can be computed
in polynomial time for Markov decision processes, and the complexity of
checking if the value is equal to a rational is in NP $\cap$ co-NP for
2$\frac{1}{2}$-player two-level discounted games.
Two-level discounted games subsume classical discounted games, and
our complexity bounds match the best known results for classical
discounted games.

Technically, we combine the existence of pure memoryless strategies in
discounted games \cite{Sha53} with the existence of pure memoryless
strategies in (undiscounted) 2$\frac{1}{2}$-player reachability games \cite{KS81,FV97,Mar98} together
with a reduction from two-level games to a one-level discounted game.
In particular, we show that for Markov decision processes, we can formulate
the value at a state as a linear programming problem over the states of
the two-level game.
Thus, the games have an ordered field property: if all constants in the definition of the game
come from a field $F$, then the value is also in $F$; in particular, games with rational probabilities
and rational discount factors have a rational value.
Together with the existence of pure memoryless strategies, this implies that the decision
problem to check if the value is equal to a given rational is in NP $\cap$ co-NP.
We also give a strategy improvement algorithm to compute the value, by combining strategy
improvement algorithms for stochastic reachability \cite{Con93} and discounted games \cite{FV97}.

Thus our new model of stochastic games provides a uniform framework for decision making across different
time scales, and our algorithms show how to decide optimally in such a framework.

\section{Definitions}
\label{section:definition}

We consider several classes of turn-based games: two-player
turn-based probabilistic games ($2\half$-player games), two-player
turn-based deterministic games ($2$-player games), and Markov
decision processes ($1\half$-player games).

\smallskip\noindent{\bf Notation.}
For a finite set~$A$, a {\em probability distribution\/} on $A$ is a
function $\trans\!:A\to[0,1]$ such that $\sum_{a \in A} \trans(a) = 1$.
We denote the set of probability distributions on $A$ by $\distr(A)$.
Given a distribution $\trans \in \distr(A)$, we denote by
$\supp(\trans) = \{x \in A \mid \trans(x) > 0\}$ the {\em support\/}
of $\trans$.

\smallskip\noindent{\bf Game graphs.}
A \emph{turn-based probabilistic game graph}
(\emph{$2\half$-player game graph})
$\gamegraph =((S, E), (\SA,\SB,\SR),\trans)$
consists of a directed graph $(S,E)$, a partition $(\SA$, $\SB$,
$\SR)$ of the finite set $S$ of states, and a probabilistic transition
function $\trans$: $\SR \rightarrow \distr(S)$, where $\distr(S)$ denotes the
set of probability distributions over the state space~$S$.
The states in $\SA$ are the {\em player-$\PA$\/} states, where player~$\PA$
decides the successor state; the states in $\SB$ are the {\em
player-$\PB$\/} states, where player~$\PB$ decides the successor state;
and the states in $\SR$ are the {\em probabilistic\/} states, where
the successor state is chosen according to the probabilistic transition
function~$\trans$.
We assume that for $s \in \SR$ and $t \in S$, we have $(s,t) \in E$
iff $\trans(s)(t) > 0$, and we often write $\trans(s,t)$ for $\trans(s)(t)$.
For technical convenience we assume that every state in the graph
$(S,E)$ has at least one outgoing edge.
For a state $s\in S$, we write $E(s)$ to denote the set
$\set{t \in S \mid (s,t) \in E}$ of possible successors.
The size of a game graph $G=((S,E),(S_1,S_2,\SR),\trans)$
is
\[
|G|= |S| + |E| + \sum_{t \in S} \sum_{s \in \SR} |\trans(s)(t)|;
\]
where $|\trans(s)(t)|$ denotes the space to represent the transition
probability $\trans(s)(t)$ in binary.

The {\em turn-based deterministic game graphs} (\emph{2-player game graphs})
are the special case of the $2\half$-player game graphs with $\SR = \emptyset$.
The \emph{Markov decision processes} (\emph{$1\half$-player game graphs})
are the special case of the $2\half$-player game graphs with
$\SA = \emptyset$ or $\SB = \emptyset$.
We refer to the MDPs with $\SB=\emptyset$ as \emph{player-$\PA$ MDPs},
and to the MDPs with $\SA=\emptyset$ as \emph{player-$\PB$ MDPs}.

\smallskip\noindent{\bf Plays and strategies.}
An infinite path, or \emph{play}, of the game graph $\gamegraph$ is an
infinite
sequence $\pat=\seq{s_0, s_1, s_2, \ldots}$ of states such that
$(s_k,s_{k+1}) \in E$ for all $k \in \Nats$.
We write $\Paths$ for the set of all plays, and for a state $s \in S$,
we write $\Paths_s\subseteq\Paths$
for the set of plays that start from the state~$s$.

A \emph{strategy} for  player~$\PA$ is a function
$\straa$: $S^*\cdot \SA \to \distr(S)$ that assigns a probability
distribution to all finite sequences $\vec{w} \in S^*\cdot \SA$ of states
ending in a player-1 state
(the sequence represents a prefix of a play).
Player~$\PA$ follows the strategy~$\straa$ if in each player-1
move, given that the current history of the game is
$\vec{w} \in S^* \cdot \SA$, she chooses the
next state according to the probability distribution $\straa(\vec{w})$.
A strategy must prescribe only available moves, i.e.,
for all $\vec{w} \in S^*$, and
$s \in \SA$ we have $\supp(\straa(\vec{w} \cdot s)) \subseteq E(s)$.
The strategies for player~2 are defined analogously.
We denote by $\Straa$ and $\Strab$ the set of all strategies for player~$\PA$
and player~$\PB$, respectively.

Once a starting state  $s \in S$ and strategies $\straa \in \Straa$
and $\strab \in \Strab$ for the two players are fixed, the outcome
of the game is a random walk $\pat_s^{\straa, \strab}$ for which the
probabilities of events are uniquely defined, where an \emph{event}
$\Aa \subseteq \Paths$ is a measurable set of paths.
For a state $s \in S$ and an event $\Aa\subseteq\Paths$, we write
$\Prb_s^{\straa, \strab}(\Aa)$ for the probability that a path belongs
to $\Aa$ if the game starts from the state $s$ and the players follow
the strategies $\straa$ and~$\strab$, respectively.
Similarly we denote by $\Exp_{s}^{\straa,\strab}(\cdot)$ the
expectation under the probability measure $\Prb_s^{\straa,\strab}(\cdot)$.
In the context of player-1 MDPs we often omit the argument~$\strab$, because
$\Strab$ is a singleton set.

We classify strategies according to their use of randomization and memory.
Strategies that do not use randomization are called pure; formally,
a player-1 strategy~$\straa$ is \emph{pure} if for all $\vec{w} \in S^*$
and $s \in \SA$, there is a state~$t \in S$ such that
$\straa(\vec{w}\cdot s)(t) = 1$.
%The pure strategies for player~2 are defined analogously.
We denote by $\Straa^{P}\subseteq\Sigma$ the set of pure strategies for
player~1.
In order to emphasize the potential use of randomization,
we call a (general) strategy \emph{randomized}.
Let $\mem$ be a set called \emph{memory}, that is, $\mem$ is a
set of memory elements.
A player-1 strategy $\straa$ can be described as a pair of functions
$\straa=(\straa_u,\straa_m)$:
a \emph{memory-update} function $\straa_{u}$: $S \times \mem \to \mem$
and a \emph{next-move} function $\straa_{m}$:
$\SA \times \mem \to \distr(S)$.
We can think of strategies with memory as input/output automata
computing the strategies (see~\cite{DJW97} for details).
A strategy $\straa=(\sigma_u,\sigma_m)$ is
\emph{finite-memory} if the memory $\mem$ is finite,
and then the size of the strategy
$\straa$, denoted as $|\straa|$, is the size of its memory $\mem$, i.e., $|\straa|=|\mem|$.
We denote by $\Straa^{\FM}$ the set of finite-memory strategies for player~1,
and by $\Straa^{\PF}$ the set of {\em pure finite-memory} strategies;
that is, $\Straa^{\PF}= \Straa^P \cap \Straa^F$.
The strategy $(\sigma_u,\sigma_m)$ is {\em memoryless\/} if $|\mem|=1$;
that is, the next move
does not depend on the history of the play but only on the current state.
A memoryless player-1 strategy can be represented as a function
$\straa$: $\SA \to \distr(S)$.
A \emph{pure memoryless strategy} is a pure strategy that is memoryless.
A pure memoryless strategy for player~1 can be represented as
a function $\straa$: $\SA \to S$.
We denote by $\Straa^M$ the set of memoryless strategies for player~1,
and by $\Straa^{\PM}$ the set of pure memoryless strategies;
that is, $\Straa^{\PM}= \Straa^P \cap \Straa^M$.
Analogously we define the corresponding strategy families
$\Pi^P$, $\Pi^{\FM}$, $\Pi^{\PF}$, $\Pi^M$, and $\Pi^{\PM}$ for
player~2.

\medskip\noindent{\bf Two-level discounted games.}
A \emph{two-level discounted game} consists of a turn-based probabilistic
game graph $G$; a partion of the state space $S$ into $(S_u,S_l)$
the set $S_u$ of upper level states and the set $S_l$
of lower level states; and a reward function $r: S_u \to \reals_{> 0}$
that maps every upper level state to a positive real-valued reward.
We also require that from every state $s \in S_l$ player~1 can ensure to
reach a state in $S_u$ with probability~1.
In other words, for all $s \in S_l$, there exists a player~1 strategy
$\straa$ such that against all player~2 strategies $\strab$ we have
$\Prb_s^{\straa,\strab}(\Reach{S_u})=1$, where $\Reach{S_u}$ is the
set of paths that visit a state in $S_u$.

\smallskip\noindent{\bf Discounted objectives.}
An objective $f$ is a measurable function $f:\Paths \to \reals$ that
assigns to every path a real-valued payoff.
The \emph{discounted objective} in two-level discounted games is a
measurable function $\TwoDisc: \Paths \to \reals$ defined as follows: for
$0 < \beta <1$, consider a path $\pat=\seq{s_0,s_1,s_2,\ldots}$ and
for an index $i \geq 0$, let
\[
\alpha(i)= \begin{cases}
0 & s_i \in S_l; \\
\beta^k \cdot r(s_i) & s_i \in S_u \text{ and the number of $S_u$ states in $\seq{s_0,\ldots,s_{i-1}}$ is $k-1$};
\end{cases}
\]
then %%\mynote{RUPAK: NORMALIZED BY $(1-\beta)$, CHECK:}
\[
\TwoDisc(\pat)= (1 - \beta) \cdot \sum_{i=0}^\infty \alpha(i).
\]
In other words, the payoff of a path is the normalized discounted sum of the rewards of
the path and the discounting is applied for every upper level state.

\medskip\noindent{\bf Optimal strategies.}
Given objectives $f$ and $-f$ for player~1 and player~2,
respectively, we define the \emph{value} functions
$\va$ and $\vb$ for the players~1 and~2, respectively, as the following
functions from the state space $S$ to the set $\reals$ of reals:
for all states $s\in S$, let
\[
\va(f)(s)=
\displaystyle \sup_{\straa \in \Straa} \inf_{\strab \in \Strab}
\Exp_s^{\straa,\strab}[f];
%%\\[1ex]
\quad
\vb (-f)(s) =
\displaystyle \sup_{\strab\in \Strab}
\inf_{\straa \in \Straa} \Exp_s^{\straa,\strab}[-f].
\]
In other words, the value $\va(f)(s)$ gives the maximal
expectation with which player~1 can achieve her objective $f$ from state~$s$,
and analogously for player~2.
The strategies that achieve the value are called optimal:
a strategy $\straa$ for player~1  is \emph{optimal} from the state
$s$ for the objective $f$ if
$\va(f)(s)=\inf_{\strab \in \Strab} \Exp_s^{\straa,\strab}[f]$.
The optimal strategies for player~2 are defined analogously.
We now state the classical determinacy results for $2\half$-player
games with measurable objectives.

\begin{theorem}[Quantitative determinacy~\cite{Mar98}]
\label{thrm:quan-det}
For all $2\half$-player game graphs $G=((S,E),(S_1,S_2,\SR),\trans)$
and for all measurable functions $f$, we have
$\va(f)(s) + \vb(-f)(s)=0$ for all states $s\in S$.
\end{theorem}

The determinacy result follows for two-level discounted games.
In the following section we will study the complexity of optimal
strategies and the computational complexity of solving two-level discounted
games.
We first recall a result about the classical discounted games.
The classical discounted games are special cases of two-level discounted games
such that $S_l=\emptyset$; i.e., the game consists of only upper level states.
We refer to this class of games as {\em one-level} discounted games.

\begin{theorem}[Memoryless determinacy of one-level discounted games~\cite{FV97}]
For all $2\half$-player one-level discounted games, pure memoryless
optimal strategies exist for both players.
\end{theorem}

\section{Strategy and Computational Complexity}

We first show that pure memoryless optimal strategies exist in two-level
discounted games.
We first present a special class of two-level discounted games and reduce
it to one-level discounted games.

\medskip\noindent{\bf One-step two-level discounted games.} The class of
one-step two-level discounted games are the special case of two-level
discounted games such that the following restrictions are satisfied:
(a)~$S_l \subseteq \SR$ (i.e., every lower level state is a probabilistic
state); and
(b)~$E \cap S_l \times S \subseteq S_l \times S_u$ (i.e., every successor of a
state in $S_l$ is a state in $S_u$).
In other words, in one-step two-level discounted games from any lower level
state the upper level states are reached in one step with probability~1.
An one-step two-level discounted game can be reduced to one-level
discounted games as follows.
We convert every state in $S_l$ to a state in $S_u$; and the reward function is
modified as follows: we add rewards to the states in $S_l$ to take care of the extra
discounting step for converting a state in $S_l$ to a state in $S_u$, i.e., for
a state $s \in S_l$ its reward is assigned as
\[
r(s) = (1-\beta)\cdot \sum_{t\in S} r(t) \cdot \trans(s)(t).
\]
Since we can reduce one-step two-level discounted games to one-level discounted games,
the existence of pure memoryless optimal strategies in one-step two-level discounted
games follows.

\begin{theorem}
Pure memoryless optimal strategies exist for both players in two-level discounted
games.
\end{theorem}
\begin{proof}
To prove the results we will use the existence of pure memoryless optimal strategies
in reachability games, and present a reduction to one-step two-level discounted games.

First, for states in $S_l$ we consider a reachability game as follows:
once the game reaches a state $s \in S_u$, then player~1 receives the payoff
$\va(\TwoDisc)(s)$ and the game stops.
The goal of player~1 is to maximize the payoff.
Since this game is a reachability game, from the existence of pure memoryless
optimal strategies in turn-based probabilistic reachability games~\cite{Con92},
it follows that pure memoryless optimal strategies $\straa^*$ and $\strab^*$
exist for both players in this game.
Since the reward function is positive, it follows that $\va(\TwoDisc)(s)>0$ for
all $s \in S_u$.
Since in two-level discounted games player~1 can ensure to reach $S_u$ with
probability~1, it follows that once $\straa^*$ and $\strab^*$ are fixed from
all states in $S_l$ states in $S_u$ are reached with probability~1.
Let $T$ denote the random time when the game first reaches a state in $S_u$,
and $\Theta_T$ denote the random variable for the $T$-th state.
The strategies $\straa^*$ and $\strab^*$ ensure the following:
\begin{enumerate}
\item for all strategies $\strab$ and for all states $s \in S_l$ we have
\[
\va(\TwoDisc)(s) \geq \sum_{t\in S_u} \Prb_{s}^{\straa^*,\strab}(\Theta_T=t) \cdot
\va(\TwoDisc(t));
\]
\item for all strategies $\straa$ and for all states $s \in S_l$ we have
\[
\va(\TwoDisc)(s) \leq \sum_{t\in S_u} \Prb_{s}^{\straa,\strab^*}(\Theta_T=t) \cdot
\va(\TwoDisc(t));
\]
and
\item $\Prb_s^{\straa^*,\strab^*}(T < \infty)=1$ and
$\va(\TwoDisc)(s) = \sum_{t\in S_u} \Prb_{s}^{\straa^*,\strab^*}(\Theta_T=t) \cdot
\va(\TwoDisc(t))$.
\end{enumerate}
We now present a reduction from two-level discounted games to
one-step two-level discounted games.
We replace each state $s \in S_l$ as a probabilistic state such that
from $s$ the successor state distribution is as follows:
$\trans(s)(t)=\Prb_s^{\straa^*,\strab^*}(\Reach{t})$ for $t \in S_u$.
The following assertions hold in the one-step two-level discounted game for
states in $S_u$:
\begin{enumerate}
\item for all states $s \in S_1 \cap S_u$, we have
\[
\va(\TwoDisc)(s) =\max_{t\in E(s)} \beta\cdot r(s) + (1-\beta)\cdot \va(\TwoDisc(t));
\]
\item for all states $s \in S_2 \cap S_u$, we have
\[
\va(\TwoDisc)(s) =\min_{t\in E(s)} \beta \cdot r(s) + (1-\beta)\cdot \va(\TwoDisc(t));
\]
and
\item ~for all states $s \in \SR \cap S_u$, we have
\[
\va(\TwoDisc)(s) = \beta \cdot r(s) + (1-\beta)\cdot\sum_{t \in S} \trans(s)(t)\cdot\va(\TwoDisc(t)).
\]
\end{enumerate}
The following assertions hold in the one-step two-level discounted game for
states in $S_l$:
\begin{enumerate}
\item for all states $s \in S_1 \cap S_l$, we have
$\va(\TwoDisc)(s) =\max_{t\in E(s)} \va(\TwoDisc(t))$;
\item for all states $s \in S_2 \cap S_u$ we have
$\va(\TwoDisc)(s) =\min_{t\in E(s)} \va(\TwoDisc(t))$;
and
\item for all states $s \in \SR \cap S_u$, we have
$\va(\TwoDisc)(s) = \sum_{t \in S} \trans(s)(t)\cdot\va(\TwoDisc(t))$.
\end{enumerate}
From the above inequalities, the classical correctness proof for one-level discounted games,
and the reduction of one-step two-level discounted games to one-level discounted
games, it follows that the values in the original game and the reduced
one-step two-level discounted games coincide.
The combination of the pure memoryless optimal strategy in the discounted
game obtained after reduction, and the pure memoryless optimal strategy in
the reachability game is a witness optimal strategy in the two-level discounted game.
Hence the result follows.
\qed
\end{proof}

\medskip\noindent{\bf Solution for two-level discounted MDPs.}
The existence of pure memoryless optimal strategies for two-level
discounted games is proved
by combining the solution of a reachability game and one-level discounted
games.
The result for MDPs follows as a special case.
The value function for MDPs can be obtained from the solution of a linear programming
problem which combines the linear programming solution for MDPs with reachability and
one-level discounted objectives.
The linear program for player-1 MDPs is as follows:
the objective function is $\min_{s\in S} x_s$ subject to the following
constraints
\[
\begin{array}{rcll}
x_s & \geq &  x_t & s \in S_l \cap S_1; (s,t)\in E; \\[1ex]
x_s & = & \sum_{t \in S} \trans(s)(t)\cdot x_t & s \in S_l \cap \SR; \\[1ex]
x_s & \geq & \beta \cdot r(s) + (1-\beta)\cdot x_t & s \in S_u \cap S_1;
(s,t)\in E; \\[1ex]
x_s & = & \beta \cdot r(s) + (1-\beta)\cdot \sum_{t \in S}
\trans(s)(t)\cdot x_t \quad & s \in S_l \cap \SR;
\end{array}
\]
The solution for player-2 MDPs is similar. This gives us the following result.

\begin{theorem}
Given a two-level discounted game on a player-1 MDP or a player-2 MDP, the
value at all states can be computed in polynomial time.
\end{theorem}

A class of discounted games
has the {\em ordered field property} if for every game $\TwoDisc$
in the class with rewards, transition probabilities, and
discount factors chosen from a field $F$, we have that the value
$\va(\TwoDisc)(s)$ is also in $F$ for each state $s$.

\begin{corollary}[Ordered field property]
Given a two-level discounted game, if the rewards, discount factor, and
transition probabilities are rational, then the value at evey state is
rational.
The class of all two-level discounted games have the ordered field property.
\end{corollary}
\begin{proof}
The results follows from the existence of pure memoryless optimal strategies
and the existence of linear program that characterizes the values on MDPs.
Once a pure memoryless optimal strategy is fixed we have an MDP, and by the
linear program characterizing the value for MDPs it follows that if the rewards,
discount factor and transition probabilities are rational, then the value at
every state is rational.
The ordered field property follows from similar arguments.
\qed
\end{proof}

\smallskip\noindent{\bf Complexity of two-level discounted games.}
Since pure memoryless optimal strategies exist for both players in two-level
discounted games, and MDPs with two-level discounted objectives can be solved
in polynomial time, it follows that the decision problem for the value
function in two-level discounted games can be solved in NP $\cap$ coNP.
Hence we have the following result.

\begin{theorem}
Given a two-level discounted game, a rational number $q$, a state $s$, and
$\bowtie \in \set{\geq,>,\leq,<,=}$, whether $\va(\TwoDisc)(s) \bowtie q$
can be decided in NP $\cap$ coNP.
\end{theorem}

\smallskip\noindent{\bf Algorithm for computing values.}
The existence of pure memoryless strategies ensure the correctness of
the following naive algorithm
to compute the values in two-level discounted game:
(a)~enumerate all pure memoryless
strategies, and for each pure memoryless strategy compute the value for the MDP
obtained by fixing the strategy (using the linear program), and
(b) choose the value of the best pure memoryless strategy.
The above algorithm is an exhaustive search on the set of pure memoryless strategies.
We now describe an efficient search on the set of pure memoryless strategies
given as a strategy improvement algorithm for two-level discounted games.
The strategy improvement algorithm combines in a hierarchical fashion two classical
strategy improvement algorithms: (a)~the strategy improvement algorithm for
stochastic games with discounted objectives~\cite{FV97} and (b)~the strategy improvement
algorithm for stochastic reachability games~\cite{Con93}.

The strategy improvement algorithm is as follows:
(a)~fix a pure memoryless strategy at the upper-level states;
(b)~apply the strategy improvement algorithm for  reachability games
for the lower-level reachability game to compute values given the strategy that is
fixed in the higher-level game; and
(c)~once the values are computed, apply the strategy improvement step for discounted
games to improve the upper-level strategy.
The algorithm stops when no improvement is possible and obtains a pure memoryless
optimal strategy.
This gives us a strategy improvement algorithm to compute values in two-level discounted
games.

\section{Conclusion}

We have introduced a new model of stochastic games that provide
a uniform framework for decision making across different time scales.
We have shown that pure memoryless optimal strategies exists in these games.
Our framework subsumes classical discounted games, and provides a natural extension
in which discounting is applied at different time granularities.
We show that in our framework the solution for MDPs can be achieved in
polynomial time matching the best known bound of MDPs with discounted
objectives.
For two-level turn-based stochastic games we show
that whether the value is equal to a rational
can be decided in NP $\cap$ coNP, matching the best known
complexity bound for discounted stochastic games.

%%\bibliographystyle{eptcs}
%%\bibliography{diss}

\end{document}